\documentclass[aps,prl,superscriptaddress,twocolumn]{revtex4-1}
\usepackage{bm} \usepackage{graphicx} \usepackage{amsmath}
\usepackage{amsthm}
\usepackage{amssymb} \usepackage{epstopdf} \usepackage{txfonts}

\newcommand{\prlsection}[1]{{\em {#1}:--~}}

\newtheorem{lemma}{Lemma}
\newtheorem{prop}{Proposition}
\DeclareMathOperator*{\Tr}{Tr}

\DeclareMathOperator*{\Sp}{Sp}

\DeclareMathOperator{\diag}{diag}
\DeclareRobustCommand\openzero{\leavevmode\hbox{0\kern-.55em0}}
\mathchardef\minus="002D

\newcommand{\ket}[1]{|{#1}\rangle}
\newcommand{\bra}[1]{\langle{#1}|}
\newcommand{\mean}[1]{\langle{#1}\rangle}
\newcommand{\sket}[1]{|{#1})}
\newcommand{\sbra}[1]{({#1}|}
\newcommand{\sbraket}[2]{({#1}|{#2})}
\newcommand\V[1]{\boldsymbol{\mathbf{#1}}}

\begin{document}

\title{Quantum information-geometry of dissipative quantum phase transitions}

\author{Leonardo Banchi}
\affiliation{Institute for Scientific Interchange Foundation, Via Alassio 11/c 10126 Torino, Italy}
\author{Paolo Giorda}
\affiliation{Institute for Scientific Interchange Foundation, Via Alassio 11/c 10126 Torino, Italy}
\author{Paolo Zanardi}
\affiliation{
Department of Physics and Astronomy, and Center for Quantum Information Science \& Technology, University of Southern California, Los Angeles, CA 90089-0484}
\affiliation{    Centre for Quantum Technologies,  National University of Singapore,  2 Science Drive 3, Singapore 117542}

\date{\today}

\begin{abstract}
  A general framework for analyzing the recently discovered
  phase transitions in the steady state of dissipation-driven open quantum
  systems is still missing.
 In order to fill this gap we extend the so-called fidelity approach to quantum phase transitions to  open systems whose steady state is a Gaussian
  Fermionic state. We endow the  manifold of correlations matrices of steady-states with a metric tensor $g$ measuring
the distinguishability  distance  between solutions corresponding to different set of control parameters.
The phase diagram can be then mapped out  in terms of the scaling-behavior of $g$ and connections with the Liouvillean gap and the model correlation functions  unveiled.
We argue that the fidelity approach, thanks to its differential-geometric  and information-theoretic nature, provides novel insights
  on dissipative quantum critical phenomena as well as a general and powerful strategy to explore them.
\end{abstract}

\maketitle

\prlsection{Introduction}
The occurrence of typical equilibrium phenomena in out equilibrium driven condensed matter systems
(e.g. long range order, topological order, quantum phase transitions)
has been recently discovered \cite{prosen2008quantum,diehl2008quantum,diehl2010dynamical,dalla2010quantum}. This poses new, fascinating and challenging problems both
at the theoretical and at the experimental level. Indeed, it has been shown that dissipation processes can in principle be controlled and tailored
in order to compete with systems free evolution and to realize fundamental protocols such as quantum state preparation \cite{kastoryano2011dissipative},
quantum simulation \cite{barreiro2011open}, and computation \cite{verstraete2009quantum}.
%An accurate description of these phenomen appears to  require the development of new theoretical tools.
The natural question that arises is whether and how the methods typically used in the equilibrium realm can be
adapted to characterize non-equilibrium problems.
In particular, the occurrence of quantum phase transitions (QPTs) in non-equilibrium steady states (NESS) which are the
results of complex many-body dissipative evolutions is far from being understood and we still lack a comprehensive and systematic framework
able to link equilibrium and  non-equilibrium properties.\\
In this Letter we propose a new information-geometric strategy for describing
NESS-QPT  based on the study of a quantity borrowed from quantum information
theory, i.e. the fidelity $\mathcal F$
between quantum states. This general approach has been so far successfully
applied to a large variety of ground state QPTs (GS-QPTs) \cite{zanardi2006, zanardi2007information, campos2007quantum, zanardi2007bures} and quantum chaos \cite{giorda2010quantum}.
In the context of NESS-QPTs the set of (control) parameters
$\lambda \in \mathcal{M}$ defines a Liouvillean superoperator $\mathcal L(\lambda)$ which drives the system, independently of the chosen initial state, to the corresponding (unique) NESS $\rho(\lambda)$. Depending on $\lambda$ the NESS can exhibit quite different properties and the system can exhibit NESS-QPTs. The main idea behind the fidelity approach is the following: when dramatic structural changes occur in $\rho(\lambda)$, e.g. approaching a critical point, the geometric-statistical distance $d[\rho(\lambda),\rho(\lambda+\delta \lambda)]$ between two infinitesimally close states grows as they become more and more statistically distinguishable. 
Although there are several metrics in information geometry 
\cite{petz2008quantum,amari2000methods,bengtsson2006geometry}  
for (mixed) density operators $\rho(\lambda)$ 
\cite{braunstein1994statistical}, 
here we concentrate on the Bures metric 
$ds^2_B = 2[1-\mathcal F(\rho, \rho+d\rho)]$. The latter is written in terms of
the Uhlmann fidelity \cite{uhlmann1976transition} $\mathcal F$, 
%$\mathcal F(\rho,\sigma)= \Tr\sqrt{\sqrt{\sigma}\rho\sqrt\sigma}$, 
and, in turns, represents the natural measure of {\em distinguishability}.
The infinitesimal distance $ds^2_B$, when expressed in terms of the parameters $\lambda$, provides a metric $g$ onto the parameters manifold $\mathcal{M}.$
The tensor $g$ is the fundamental tool of the fidelity approach: it has been shown that the study of its scaling behaviour (extensive vs. superextensive) allows a systematic study of GS-QPTs \cite{campos2007quantum,polko2013}. 
%and it is therefore the natural candidate for the investigation of NESS-QPTs.

%In this paper we apply the fidelity approach to NESS-QPT. 
Dissipative QPTs are of a different nature of the standard 
QPTs at zero temperature.
Accordingly, in spite of some obvious yet somewhat superficial similarity, 
their understanding calls for  a different set of conceptual as well as 
mathematical tools. In the first place, stationary states are the result
of an equilibration process: NESS-QPTs needs a new equilibration
time after the perturbation and, as such, they are not a result of an 
{\it adiabatic} reorganization of the (ground) state.
From a mathematical point of view, a NESS
is the zero eigenvalue {\em density matrix} of the 
{\em non-hermitean} Liouvillean superoperator ${\cal L} $,  as opposed
to pure eigenvectors of an Hermitian Hamiltonian operator $H$.
This implies that, on the one hand, one has to employ the more sophisticated 
information-geometry of mixed states and, on the other hand, that the whole 
wealth of powerful results stemming out of Hermiticity, 
e.g. spectral theorem and perturbation theory, are in the dissipative 
case simply not available. 
The challenge is here to find out a suitable way to parametrize 
the manifold of stationary states
and pull-back into the parameter manifold the state metric. 
This is in general a quite daunting task, but restricting to the physically
relevant case of quadratic Liouvillean can be achieved. 
Specific models belonging to this class indeed display rich non-equilibrium features and NESS-QPTs, which have been characterized by studying long range magnetic correlations (LRMC) and the Liouvillean spectral gap $\Delta_{\mathcal{L}}$ \cite{prosen2008quantum,horstmann2013noise}.

%We concentrate on an important set of dissipation driven Fermionic systems 
%whose NESS are Gaussian Fermionic (GF) states \cite{prosen2008third}. 
We derive a general formula for the Bures distance over the set of 
Gaussian Fermionic (GF) states and the metric tensor $g$ over the
parameter manifold. Then
we discuss how the scaling of the metric implies both the closing of 
$\Delta_{\mathcal L}$ and the divergences of some two-point correlations.
Finally we apply our theoretical framework to exactly solvable models. 
Our analysis  demonstrates  that the NESS phase diagram can be accurately mapped by studying  the (finite-size) scaling behaviour of the metric tensor $g;$
 critical lines can be identified and the different phases distinguished.

\prlsection{Bures metric for Gaussian Fermionic states}
The calculation of the Bures distance is a notoriously hard task for large
Hilbert spaces: 
standard methods \cite{braunstein1994statistical} are computationally not
applicable for many-body systems and finding an efficient way to evaluate
$ds_B^2$ is still a subject of active research \cite{ercolessi2013symmetric}.
Here we show a compact and efficient way to evaluate the Bures metric 
(for convenience we use a rescaled metric $ds^2=8\,ds^2_B$)
when the state space is restricted to the physically 
important case of Gaussian Fermionic states.
%The first step of our information-geometric analysis is based on the pull-back
% of  the Bures metric for GF-states onto the manifold of the two point correlation functions.
Consider a system of $n$ Fermion modes described by a set of $2n$ Majorana
operators $w_i$. These operators are Hermitian, linearly depend on the
Fermionic creation and annihilation operators via
$w_\ell=f_\ell+f_\ell^\dagger$, $w_{n+\ell} = i(f_\ell-f_\ell^\dagger)$,
$\ell=1\dots n$, and satisfy the algebra
$\{w_i, w_j\}=2\delta_{ij}$.
%We define the (non-canonical) ``vectorising"  isomorphism $\phi\colon M_{2n}({\mathbf{C}})\rightarrow  ({\mathbf{C}}^{2n})^{\otimes\,2}\,/\, |n\rangle\langle m|\rightarrow |n\rangle\otimes|m\rangle.$ This is also a Hilbert-space
%isomorphism, namely $\langle \phi(A),\phi(B)\rangle = \langle A,\,B\rangle={\rm{Tr}}\,(A^\dagger B).$
A GF-state $\rho$, i.e. a Gaussian state in terms of the operators $w_j$, 
%Owing to the Wick's theorem, this state 
is completely specified by the two-point correlation functions
$C_{ij}=\frac12\langle[w_i,w_j]\rangle_\rho$, where
the complex $2n\times 2n$ matrix $C$ is imaginary
and anti-symmetric.
%The following theorem shows that the fidelity metric
%(for convenience we use a rescaled metric $ds^2=8\,ds^2_B$)
%can be expressed in terms of the correlation matrix $C$.

%\begin{thm}\label{fidelity}
With this natural parametrization the metric can be pulled back from the
many-body Liouville space to the manifold of the two point correlation 
functions. Indeed, in the Supplementary Material (SM) we have shown that 
the fidelity metric around the GF-state $\rho$ specified
by the correlation function $C$ is given by
\begin{equation}
ds^2=\Tr\left[dC(\openone-{\rm{Ad}_C})^{-1}dC\right]
\ =:\|(\openone-{\rm{Ad}_C})^{-\frac12}\,dC\|^2_2
\label{e.bures}
\end{equation}
where ${\rm{Ad}_C}X{:=} C X C^\dagger{=}CXC$ is the adjoint
action and $^{-1}$ refers to the pseudo-inverse.
In particular, when $\rho$ is pure, $\Sp(C)=\{\pm 1\}$ and the above
equation reduces to $ds^2_{pure}{=} \|dC\|_2^2/2$.
%\end{thm}

%If $C$, and thus $\rho$, depends on a set of external parameters
%$\{\lambda_\mu\}$, then $dC=\sum_\mu d\lambda_\mu\partial_\mu C $ and
This is {\em per se} an interesting novel result
 %(whose scope may well go beyond the applications discussed in this paper) 
but it is just the first 
step of our analysis.
In fact the crucial physical information is contained in 
the external parameters  $\{\lambda_\mu\}\in\mathcal M$ of the
model. As $dC=\sum_\mu d\lambda_\mu\partial_\mu C $ we obtain 
\begin{align}
  ds^2 &= \sum_{\mu,\nu}g_{\mu\nu}d\lambda_\mu d\lambda_\nu~,&
g_{\mu\nu}&=\sum'_{rs}\frac{(\partial_\mu C)_{rs} (\partial_\nu C)_{sr}}{1-c_r c_s} ~,
\label{e.buresbasis}
\end{align}
where $C = \sum_{r} c_r\ket r\bra r$, with $c_r\in\mathbf{R}$ and
$(\partial C)_{rs}=\bra r\partial C\ket s$, i.e. the sum in the above equation
is performed in the basis in which $C$ is diagonal and it is restricted over
the elements such that $c_r c_s\neq 1$.
The infinitesimal distance $ds^2$ encodes the statistical distinguishability between two infinitesimally close Gaussian Fermionic states; this result is
completely general and it can be used to study the geometrical properties of manifolds of GF-states.
Eqs.\eqref{e.bures} and \eqref{e.buresbasis} provide the basic tool for studying the phase transitions occurring when the NESS are GF-states.
%For GS-QPTs a superextensive behaviour of $ds^2$ implies criticality \cite{campos2007quantum}.
%A first qualitative indication that an analogue result may hold for NESS-QPT 
In this respect, a first qualitative indication that the scaling behaviour 
of the metric can spot QPTs 
is suggested by the following inequality (see SM):
%\begin{align} ds^2 & %\le \|(1+C^{\otimes 2})^{-1}\|_\infty\,\|d\tilde C\|
%  \le 2n\, P_C\,\|dC\|_\infty^2~,
%  & P_C&= \|(\openone+C^{\otimes 2})^{-1}\|_\infty~,\label{e.dCbound}
%\end{align} where
$ds^2 \le 2n\, P_C\,\|dC\|_\infty^2$, where 
$P_C= \|(\openone+C^{\otimes 2})^{-1}\|_\infty$ and 
$\|A\|_\infty$ refers to the maximum singular value of $A$.
If %$\rho$ is mixed
$P_C{=}O(1)$
a superextensive behaviour of $ds^2$ implies some sort of singularity in the correlation functions that may reflect the occurrence of a phase transition.
%In the following we analyse the NESS-QPTs
%occurring in quadratic Fermionic models.

\prlsection{Dissipative solvable model}
We consider a Markovian dissipative open quantum system
evolution \cite{breuer2002theory} governed by the Lindblad master equation
\begin{align}
  \frac{d\rho}{dt} = \mathcal L \rho := -i[\mathcal H, \rho]
  +\sum_\mu\left(
  2L_\mu\,\rho\,L_\mu^\dagger -\{L_\mu^\dagger\,L_\mu, \rho\}\right)~,
  \label{e.lind}
\end{align}
with a quadratic Hamiltonian $\mathcal H=\sum_{ij} H_{ij} \,w_iw_j$ and
linear Lindblad operators $L_\mu = \sum_i \ell_{\mu i} \,w_i$, where
the matrices $H$ and $\ell$ depend on the parameters
$\V\lambda\in\mathcal M$ defining the specific model.
In the following we obtain the steady state $\Omega$, namely the state for
which $d\Omega/dt = \mathcal L\Omega =0$, and pull back the set of admissible
NESS to the parameter manifold. 
%We use the notation
%$\sket{\V s}:=\prod_j w_j^{s_j}$, ($s_j\in\{0,1\}$)
%for referring to the elements of $4^n$-dimensional operator spaces 
%normalized with respect to the Hilbert-Schmidt inner product: 
%$\sbraket{\V s}{\V s} \equiv \Tr[s^\dagger s]= 1$.
The Liouvillean %$\mathcal L\colon \mathcal R\rightarrow \mathcal R$
can be written as a quadratic form in terms of the following
set of $2n$ creation and annihilation superoperators
\begin{align}
a^\dagger_j \,\cdot = -\frac{i}2 \,W\,\Big\{w_j, \cdot\Big\}~,   &&
a_j \,\cdot = -\frac{i}2 \,W\,\Big[w_j, \cdot\Big]~,
\label{e.super}
\end{align}
where $W=i^n\prod_{j=1}^{2n} w_j$ is a Hermitian idempotent operator
which anti-commutes with all the $w_j$. 
A direct calculation proves that
the operators defined in Eq.~\eqref{e.super} satisfy the canonical
anti-commutation relations (CAR), $\{a_j^\dagger, a_k\} = \delta_{jk}$, and
that
$%\begin{align}
  \mathcal L = -\sum_{ij} \left(X_{ij}\,a_i^\dagger a_j +
  Y_{ij}\,a_i^\dagger a_j^\dagger/2\right)
$, %  \label{e.quadlind}
%\end{align}
where $X=4(iH+\Re M) \equiv X^*$, $Y=-8i\Im M \equiv -Y^*\equiv -Y^T$,
with $M_{ij}=\sum_\mu \ell_{\mu i} \ell_{\mu j}^* \equiv M^\dagger$.
This result was derived in \cite{prosen2008third},
but thanks to our definition \eqref{e.super},
complex transformations \cite{zunkovic2010explicit}
for unifying the different parity sectors are avoided.
The {\it two-point} correlation functions in the steady state,
$C_{ij} = \langle[w_i,w_j]\rangle$, are obtained from the
solution of the following Sylvester equation \cite{zunkovic2010explicit}
\begin{align}
  X\, C + C\,X^T = Y~.
  \label{e.C}
\end{align}
As shown in the SM the %anti-symmetric imaginary 
matrix $C$ also plays a
central role in the diagonalization of the Liouvillean. In order to simplify
our analysis we assume the real matrix $X$ to be diagonalizable,
i.e. $X=UxU^{-1}$ for $x=\diag(\{x_i\})$, $x_i\in{\bf C}$, as this condition
is always satisfied in our numerical simulations; %whereas 
%Nonetheless, our results do not depend on this assumption and 
the general (non-diagonalizable) case is discussed in the SM. 
% One can show that 
The transformation
${\V d} = U^{-1}\left({\V a}+C{\V a}^\dagger\right)$,
${\V d}^\times = U^T{\V a}^{\dagger}$,
realizes a non-unitary Bogoliubov transformation
%(NuBT), i.e. the operators $d_i$ and $d_j^\times$ satisfies the 
%CAR-algebra but $d_j^\times \neq d_j^\dagger$, 
and brings $\mathcal L$ to the diagonal
form $\mathcal L = -\sum_k x_k\, d_k^\times d_k$. 
The (unnormalized) 
steady state $\Omega$ is then obtained as the ${\V d}$-vacuum,
($d_i \Omega = 0$, $\forall j=1,\ldots2n$), i.e.
%of the diagonal operators, i.e. $d_j \sket\Omega=0$, and obtain its operator form from the NuBT.
%From the operator form of the Bogoliubov transformation, using textbook results, we found 
\begin{align}
  \Omega = %\mathcal V\sket{\V 0} =
  e^{-\frac12 {\V a}^\dagger\cdot C{\V a}^\dagger}\!({\openone})~.
  \label{e.ss}
\end{align}
where the identity operator is the $\V a$-vacuum.
%Owing to the explicit form of the superoperators \eqref{e.super},
%in the SM we show that the above state is a GF-state and that its two
%point correlation functions $\langle[w_i,w_j]\rangle_\Omega$
%are given by $C_{ij}$, i.e. by the solution of \eqref{e.C}.
The physical conditions for the existence and uniqueness of the steady
state are given in \cite{prosen2010spectral}:
if $\Delta := 2\min_i \Re(x_i) > 0$ then the solution of \eqref{e.C} is
unique and every initial state converges for
$t{\to}\infty$ to the unique steady state \eqref{e.ss}. The gap $\Delta$
represents both the inverse of the time-scale for reaching the steady state
and the gap of the Liouvillean:
  $ \min \{|\sum_{j} x_j n_j| \colon n_j {\in} \{0,1\}~\} \equiv \Delta$.\\
If $\Delta > 0$ the steady state  ${\Omega(\V\lambda)}$  is unique and, since  $\mathcal L$ smoothly depends on the parameters $\V\lambda\in\mathcal M$, it is smooth function of $\V\lambda$ \cite{magnus1985differentiating}.
If the gap $\Delta(n){\to} 0$ for $n{\to}\infty$
%the steady state is asymptotically becoming either degenerated or with a continuous part
%of the spectrum above it. %(as opposed to the gap for finite $n$).
%This is precisely what \textit{can make}
the steady state ${\Omega(\V\lambda)}$ may become a {\em non-differentiable} function of
$\V\lambda$.
However, NESS-QPT are {\it not} defined by the closing of the Liouvillean gap.
%At variance with the Hamiltonian gap in GS-QPTs, a vanishing $\Delta$ is not sufficient for the emergence of non-analyticities, and thus criticality. 
Nevertheless, the scaling properties of $\Delta(n)$ have been used as indicators of NESS-criticality 
\cite{prosen2010exact,vznidarivc2011solvable,horstmann2013noise,cai2013algebraic}.
Motivated by this, we derived in SM the following upper bound which relates the behaviour of $\Delta(n)$  and $ds^2$:
\begin{align}
  \frac{ds^2}n \le 2\frac{P_C}{\Delta^2}\,
  \left(\|dY\|_\infty+2\|dX\|_\infty\right)^2~.
  \label{e.dsbound}
\end{align}
The latter is the dissipative analogue of the GS-QPT one given in \cite{campos2007quantum}, where it was shown that superextensivity of $ds^2$
implies the closing of the Hamiltonian gap  \cite{campos2007quantum} and the occurrence of criticality. Here the bound intriguingly links the geometric quantity $ds^2$ to the dynamical property $\Delta$, and it provides the following information: if the numerator of the RHS in \eqref{e.dsbound} is $O(1)$  then
any superextensive behaviour of $ds^2 = O(n^{\alpha+1}), \, \alpha>0$ implies that the Liouvillean gap $\Delta$ closes
at least as $O(n^{-\alpha/2})$.
%It is important to stress that $\Delta$ is a completely different quantity from the Hamiltonian gap whose inverse bounds
%the geometric tensor $g_{\mu\nu}$ for GS-QPT, A complete understanding of the existence of a connection between the Liouvillean
%gap and the Hamiltonian gap ruling ground state QPT is still missing.
%e.g. in the non-dissipative case $\Sp(X)$ is imaginary and $\Delta\equiv 0$.
%Moreover, unlike in GS-QPT \cite{HastingBho}, %where, thanks to the Hasting's theorem [?],
%super-extensivity is  a {\em sufficient}
%condition for ($T=0$) criticality,
%in the dissipative case
%$ds^2=O(n^{1+\alpha}), \,(\alpha>0)$ implies just $\Delta=O(n^{\alpha/2})$, but here
%closure of gap is
%apparently neither implied by criticality nor imply it.
Therefore the geometric properties of the NESS manifold set the minimal
time scales for the reaching of the steady state.
In the next sections we specialize our results to particular solvable instances
of \eqref{e.lind} and we perform numerical and analytical analyses aiming at validating the importance and usefulness of the fidelity approach to NESS-QPT and at comparing the scaling properties of $\Delta$ and $ds^2$.

\prlsection{Boundary driven XY spin chain}
We now concentrate on a solvable spin-$\frac12$ model exhibiting a NESS-QPT 
\cite{prosen2008quantum}. Coherent interactions are described
%apply the theoretical framework developed in the previous section for analysing
%the recently discovered \cite{prosen2008quantum} NESS-QPT
%non-equilibrium phase transition in spin-$\frac12$ chains modelled 
by the XY Hamiltonian
\begin{align}
  H=\sum_{i=1}^{n-1}\left(\frac{1+\gamma}2\sigma_i^x\sigma_{i+1}^x+
  \frac{1-\gamma}2\sigma_i^y\sigma_{i+1}^y\right)+h\sum_{i=1}^n\sigma_i^z~,
  \label{e.xy}
\end{align}
where $\sigma^\alpha_j$ are the Pauli operators acting on the $j$-th spin.
The two boundary spins of the chain are coupled to two (thermal) reservoirs
via the Lindblad operators $L^{\pm}_L=\sqrt{\Gamma_L^{\pm}}\sigma^\pm_1$,
$L^{\pm}_R=\sqrt{\Gamma_R^{\pm}}\sigma^\pm_n$, where
$\sigma^\pm_j = (\sigma^x_j+i\sigma_j^y)/2$, and the strengths
$\Gamma_{L,R}^\pm$ depends on the reservoirs parameters as well on their
temperature \cite{zunkovic2010explicit}.
Owing to the Jordan-Wigner transformation, such a model can be
exactly described by a quadratic Majorana master equation \eqref{e.lind}.
The steady state of the resulting dissipative Markovian evolution is
therefore Gaussian and different phases can be identified depending
on the parameters $(h,\gamma)$ of the Hamiltonian \eqref{e.xy}. Along the lines
$h=0$, $\gamma=0$, and for $h > h_c = \vert1-\gamma^2\vert$, magnetic
correlations are short-ranged (SRMC), i.e.  the correlation
functions $C^{zz}_{ij} = \mean{\sigma_i^z\sigma^z_j}$ exhibits an exponential
decay, $C^{zz}_{ij}\approx e^{-\vert i-j \vert/\xi}$ with a localization length
$\xi \approx \sqrt{2h_c/(h-h_c)}/8$.
On the other hand, for $h<h_c$ a phase with
long-range magnetic correlations (LRMC) emerges which is characterized by
non-decaying structures in $C^{zz}_{ij}$ and a strong sensitivity to small
changes of the parameters. Around the critical point $h_c$
one finds a power-law behaviour $C^{zz}_{ij}\approx\vert i-j\vert^{-4}$.

\begin{table}[!t]
  \centering
  \begin{tabular}[t]{|c|c|c|c|c|}
    \hline
    {\bf Phase} & {\bf Parameters} %{$\mathbf h, \V\gamma$}
    & {$\mathbf{\quad\Delta\quad}$} &
    { $\quad\mathbf{|g|}\quad$} & {\bf Quality of fit}
    \\ \hline
    Critical (*) & $h = 0$ & $n^{-3}$ & $n^6$ & good
    \\ \hline
    Long-range & $0< |h| < h_c$ & $n^{-3}$ & $n^3$ & average
    \\ \hline
    Critical & $|h|\approx h_c$ & $n^{-5}$ & $n^6$ & bad
    \\ \hline
    Short-range & $|h| > h_c$ & $n^{-3}$ & $n$ & good
    \\ \hline
    \hline
    Critical (*) & $~\gamma=0,|h| < h_c$ & $n^{-3}~$ & $n^2$ & good
    \\ \hline
  \end{tabular}
  \caption{Scaling analysis of the gap $\Delta$ and of the maximum eigenvalue
    of the fidelity metric $g_{\mu\nu}$.
    %and on the variation $g_{hh}$ along the direction $h$.
    These laws does not depend on the
    particularly chosen rate $\Gamma_{L,R}^\pm$. (*) The lines
    $h=0$ and $\gamma=0$ consists
    of a SRMC region embedded in the LRMC phase; one finds
    (see discussion in the text) $|g|\approx g_{hh}$ for $h=0$ and
    $|g|\approx g_{\gamma\gamma}$ for $\gamma=0$.
  }
  \label{t.scaling}
\end{table}

In Table~\ref{t.scaling} we summarize the scaling analysis performed.
Our  results show that the
Liouvillean gap and the metric encode different information.
Indeed, unlike the Hamiltonian gap ruling ground state QPT, the Liouvillean
gap $\Delta$ closes for $n\to\infty$ both at the critical point and for
$h\neq h_c$, both in the LRMC and SRMC phase .
As the reservoirs acts only at the boundaries of the spin chain
the eigenvalues $x_k$ of the matrix $X$ for $n\gg 1$ are a small perturbation
of the $n{\to}\infty$ case where $x_k=\pm4i\omega_k$, being
$\omega_k = \sqrt{(\cos k-h)^2+\gamma^2\sin^2k}$
the quasi-particle dispersion relation of
the Hamiltonian \eqref{e.xy}. In particular $x_k$ gains a small real part and
one finds a gap $\Delta =O(n^{-3})$ for $h\neq h_c$ and $\Delta = O(n^{-5})$
for $h=h_c$.
Therefore the scaling of the Liouvillean gap allows one to identify
the transition form the SRMC phase to the LRMC phase only along the
critical line $h=h_c$, while the transition
occurring  at the $h=0$ (or $\gamma=0$) line can only be appreciated by
evaluating the long-rangeness of the magnetic correlations.
The question that naturally arises is how the different phases and transitions 
can be precisely characterized in a way similar to what happens for
GS-QPTs.
%, and in particular whether the Liouvillean gap and its scaling behaviour,
%can be considered as a good indicator of the presence of QPT in the dissipative scenario.
%The closing of the Liouvillean gap appears to be not directly related to the appearance of quantum critical points, and might not be the most relevantquantity for characterising NESS-QPT. Indeed, unlike the Hamiltonian gap ruling ground state QPT, the Liouvilleangap $\Delta$ closes for $n\to\infty$, although with different exponent, both around the critical point and for $h\neq h_c$, both in the long-range and short-range regions.
This question becomes more compelling if one compares the above results with the scaling of the geometric tensor  $g_{\mu\nu}$, and in
particular of its largest eigenvalue $|g|$, see Table~\ref{t.scaling},
and Fig.~\ref{f.scaling} for specific values of the parameters.\\
%Motivated by success of the fidelity approach in describing different kinds of phase transitions, we have applied Theorem~\ref{fidelity} to NESS-QPT and
% we have performed a scaling analysis of the largest eigenvalue (called $|g|$) of the geometric tensor $g_{\mu\nu}$ Eq.~\eqref{e.buresbasis}, and of the particular element $g_{hh}$; the results, obtained for different values of the Hamiltonian and Lindblad parameters, are summarized in Table~\ref{t.scaling}.
A first important result is that the tensor $g$ is able to identify the
transitions between SRMC and LRMC phases.
%and to make a clear distinction between different type of transitions.
On the "transition lines" $h=0$ and
$h=h_c$ one has that $|g|=O(n^6)$, while in the rest of the phase diagram
$|g|<O(n^6)$. Furthermore, a closer inspection of the elements of $g$ shows that
while $g_{hh}(h=0, \gamma) = O(n^6)$,  one has that
$g_{\gamma\gamma}(h=0,\gamma) = O(n)$:
the scaling is superextensive only if one moves away from the
line $h=0$ ($g_{hh}$) and enters in the LRMC phase, while if one moves along
the $h=0$ line ($g_{\gamma\gamma}$) i.e., if one remains in the SRMC phase,
the scaling is simply extensive and it matches the scaling displayed in
the other SRMC phase $h > h_c$.
On the other hand, the transition occurring at $\gamma=0$  
has a different scaling: %when moving toward the SMRC one has that 
$g_{\gamma\gamma} = O(n^2)$ while $g_{hh}\approx 0$. 
These findings can be further confirmed by a detailed study \cite{ToBePublished} based on the analytical results available for $\gamma \ll 1$ \cite{zunkovic2010explicit}.
It turns out that the introduction of the magnetic field or the anisotropy 
drives different transitions whose specificity is accounted for by the
different superextensive scalings. 
%It is an open question whether different superextensive scalings
%signal different kind of transitions and criticality, i.e. if the exponent 
%defines an universality class. 
%%%%%%%%%%%%%%%%%%%%%%%%%%%%%%%%%%%%%%%%%%%%%%%%%%%%%%%%%%%%%%%%%%%%%
%% : moved to the ``open questions part''
%Whether the different exponents can
%define an universality class is still an open question.
%class. 
%This question becomes more fashinating by noting that, in the considered 
%XY model, different type of symmetries (discrete vs. continuous) 
%are broken when moving away from the $h=0$ or $\gamma=0$ line.
\\
Another important result shown in Table~\ref{t.scaling} is that the metric
tensor is able to signal the presence of long-range correlations:
within the LRMC phase $ds^2$ scales superextensivity as $|g|=O(n^3)$,
and this superextensive behaviour is different from that
displayed at the transition lines.
%Indeed, the presence of long-range correlations make the steady state very sensitive to small perturbations of the external parameters \cite{zunkovic2010explicit,prosen2008quantum}, and therefore NESS corresponding to infinitesimally different values are likely to be highly statistically distinguishable, their difference is therefore captured by the fidelity, and this allows for a clear distinction between the SRMC and LRMC phases.
One is therefore led to conjecture that whole LRMC phase have a critical character, due to the presence of long range correlations.

\begin{figure}[!t]
  \begin{center}
    \includegraphics[width=.5\textwidth]{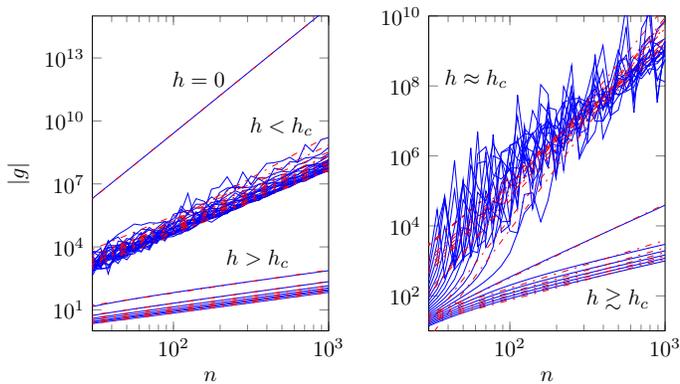}
  \end{center}
  \caption{Scaling of $|g|$ for $\gamma=0.6$ and $h\in[0,0.8]$ (left) and
    for $\gamma=0.5$ and $h\in[0.735,0.755]$ (right). 
    In both cases $\Gamma_L^+ = 0.3$,
    $\Gamma_L^- = 0.5$, $\Gamma_R^+ =0.1$, $\Gamma_R^-=0.5$.  
    Blue curves represent the numerical data,
    while red lines are linear fits, whose results are
    summarized in Table~\ref{t.scaling}.
    $|g|$ slightly fluctuates as a function of $n$
    in the LRMC phase and the relative
    amplitude of the fluctuations increases close to the critical field $h_c$.
    Due to finite size effects and to the differential nature of the
    geometric tensor, the value where $|g|$ takes its maximum is slightly
    smaller than $h_c$, and this difference depends on $n$. 
    %Thus special
    %care has to be taken for fitting the data for $h\lesssim h_c$.
  }
  \label{f.scaling}
\end{figure}

\begin{figure}[!t]
  \begin{center}
    \includegraphics[width=.5\textwidth]{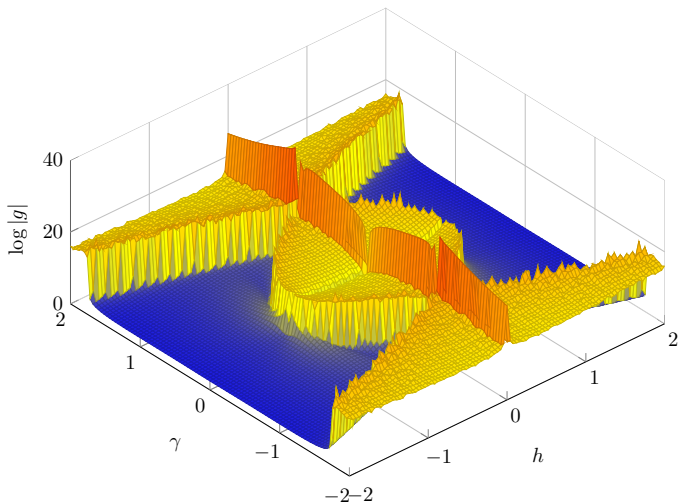}
  \end{center}
  \caption{Maximum eigenvalue $|g|$ of the fidelity metric
    \eqref{e.buresbasis} for $n=250$.
    The Lindblad parameters are the same of Fig.~\ref{f.scaling}.  The larger
    value of $|g|$ close to the phase transition line $h=h_c(\gamma)$ is not
    evident in Fig.~\ref{f.g} because of the numerical mesh and because,
    the actual values of $|g|$ for $h\approx h_c$
    can be comparable to those of the LRMC phase, depending on $n$
    (e.g. see Fig.~\ref{f.scaling}).  The qualitative form of Fig.~\ref{f.g}
    is not affected by different values of the Lindblad parameters
    $\Gamma_{L,R}^\pm$ and by the dimension $n$.
  }
  \label{f.g}
\end{figure}

The findings discussed in the above demonstrate  that the metric tensor $g$, being directly linked to the correlations properties of the Gaussian NESS, encodes  all the relevant information about the dissipative phase transition featured  by the model \eqref{e.xy}; in particular,
the specificity of the different phases (SRMC vs LRMC), and the information about the physical relevant parameters, being them the magnetic field or the anisotropy, that drive the different transitions are properly accounted for.
%This fact can be further confirmed by a detailed study 
%\cite{ToBePublished} of the bound \eqref{e.dCbound} based
%on the analytical formula for $C$ that holds for $\gamma \ll 1$ \cite{zunkovic2010explicit}.
%
As shown in Fig. \eqref{f.g}, the complete phase diagram can indeed be
reconstructed with the study of the single function $g$.
While these results are specific to the model examined,  the connection established in
\eqref{e.buresbasis} roots the behaviour of $g$ in the correlations properties
of the general class of GF-states. Accordingly, one expects the fidelity approach to have
a broader scope of application.
We would like to stress that there are compelling  questions that are still  unanswered.
In the first place the relation between $g$ and
other relevant quantities that have been used so far to characterize NESS-QPT;
For the model \eqref{e.xy},  these are the range of  correlations,
and the finite-size scaling of Liouvillean gap $\Delta.$
The latter does not entirely capture the criticality phenomenon, and
further investigation of the relation between criticality in NESS-QPT
and geometrical and dynamical aspects is in order \cite{poulin2010lieb}.
%indeed,
%The latter doesn't seem to entirely capture the criticality phenomenon, opening
%the quest for a better definition of criticality in NESS-QPT; furthermore,
%while the bound \eqref{e.dsbound} establishes a direct link
%between $g$ and $\Delta$, in the case of the specific model analized, the
%scaling of $\Delta$ hides the behaviour of $g$ with the system's size, and a
%tighter bound would be desirable.
Notice also, that, in the   XY model, different type of symmetries (discrete vs. continuous)  are broken  moving away from the $h=0$ or $\gamma=0$ line.
%For example, it 
It would be interesting to 
%find a direct connection between criticality and superextensivity of $ds^2$, and to 
understand whether the  scaling exponents of $ds^2$ at different lines can be related to different non-equilibrium universality classes. 
%The latter question becomes more fascinating by 
Extending the present results to non-Gaussian states 
\cite{de2013power} and transitions \cite{prosen2010long} is also an important
future direction.

\prlsection{Translationally invariant case}
In order to support the generality of the geometric approach in understanding
dissipative phase transitions 
we apply our theoretical framework to a different dissipative
model, first introduced in \cite{horstmann2013noise}. 
We consider an XY spin chain on a ring where each site is coupled
to the environment via 
$L_i^+ {=} \sqrt{\Gamma^+} f_i^\dagger$, $L_i^- {=} \sqrt{\Gamma^-} f_i$.
The {\it closed boundary conditions} and the {\it uniform} interaction with the
environment make the phase diagram very different from the previous
one. Indeed, in this particular translationally invariant case, the critical
points match the known values for GS-QPT: for $\gamma\neq 0$ there is a critical
field $h=1$, while in the XX case the whole segment $|h|<1$ is critical. 
In the SM we have proved that the
$|g|=O(n^2)$ for the critical values and $|g|=O(n)$ elsewhere.
The information-geometric content of this dissipative 
phase transition is not as rich as the one in Table.~\ref{t.scaling}, 
and again the scaling of the metric  tensor allows one a precise mapping of the phase-diagram.

\prlsection{Conclusions}
In this Letter we developed  an information-geometric  framework for studying
dissipative critical phenomena exhibit by the non-equilibrium steady states of Markovian evolutions described by quadratic Fermionic
Liouvillean. 
%Our strategy  is an extension of the Fidelity approach developed for zero-temperature quantum phase transitions \cite{zanardi2006,zanardi2007information,campos2007quantum}.
We first derived  a general formula for the infinitesimal Bures distance
between Gaussian Fermionic (mixed) states. This in turn allows one to define
a metric tensor $g$ on the manifold of steady states corresponding to different sets of control parameters.
The intuitive idea underlying is that a transition between two structurally different phases should be reflected by the
statistical distinguishability of pairs of infinitesimally close steady states.
The method does not require the knowledge or the existence of any order parameters,
as the tensor $g$ is directly connected to the two-point correlation functions which define the Gaussian Fermionic steady states.
We have shown that a superextensive behaviour of the tensor $g$, implies some singularity for $n{\to}\infty$ in the derivative of the correlation functions.
We have applied the method to specific (XY) models and
shown that the scaling of the geometric tensor enables one to identify both the critical lines and
to distinguish between different phases characterized by short or  long ranged correlations.
The metric tensor  encodes  also for the direction  of maximal distinguishability in the parameter
manifold, thus allowing a detailed study of the sensitivity of the steady state to small variations of some control parameters.
 This is a crucial point for  experimental applications of dissipative evolution.
%
%Moreover, we have shown that the role of the Liouvillean gap $\Delta$, is not completely clear. The bounds derived allow to infer that a superextensive behaviour of $ds^2$ implies $\Delta{\to}0$ in the thermodynamic limit, but our numerical simulations suggest that the scaling of the gap does not provide unquestionable information about criticality.
%While many intriguing questions remain open, 
The scope  of the  information-geometric approach  extends well beyond the important quadratic case analzyed in this paper and may pave the way to the
systematic study of general non-equilibrium critical phenomena. This in turn would allow the investigation of a broad class of systems
and processes which are natural candidates for the preparation of desired quantum states and realization of quantum protocols.

\prlsection{Acknowledgements}
P.Z. was supported by the ARO MURI grant W911NF-11-1-0268 and by NSF grant numbers PHY- 969969 and PHY-803304.

%\bibliography{ref}

\clearpage
\begin{center}
\Large\bf\sc Supplementary Material
\end{center}

\section{Proof of Eq.~(\ref{e.bures})}
We consider a Gaussian Fermionic state written in the following form
\begin{align}
\rho = e^{-\frac{i}4 \sum_{ij} G_{ij}\,w_i w_j}/Z~,
\label{e.gaussian}
\end{align}
where the matrix $G$ has to be real and antisymmetric.
Accordingly $G$ can be cast in the canonical form by an orthogonal matrix $Q$,
i.e.
\begin{align}
G &= Q^T\, \bigoplus_{k=1}^n
\begin{pmatrix} 0&g_k\\-g_k&0 \end{pmatrix}
  \, Q & Q^T=Q^{-1}~,
  \label{e.Q}
\end{align}
and has eigenvalues $\pm ig_k$. Moreover let $z_i = \sum_j Q_{ij} w_j$
be the new Majorana operators. Hence
\begin{align}
\rho &= \frac1Z\prod_k \left[\cosh\left(\frac{g_k}2\right) -i
\sinh\left(\frac{g_k}2\right)\,z_{2k-1} z_{2k}\right]~,
\\
Z &= \prod_k 2\cosh\left(\frac{g_k}2\right) =
\sqrt{\det\left[2\cosh\left(i\frac{G}2\right)\right]}~,
\label{e.gaussianZ}
\end{align}
where we used the fact that the eigenvalues of $iG$ are $\pm g_k$.
As $C_{ij}=\frac12\mean{[w_i, w_j]}=
\frac{2i}Z\frac{\partial Z}{\partial G_{ij}}$ one can show that
\begin{align}
C = \tanh\left(i\frac{G}2\right)~.
\label{e.corrG}
\end{align}
The correlation matrix $C=C^\dagger=-C^T$ is diagonal in the same basis of $G$
and its eigenvalues read $c_k = \tanh(g_k/2)$. Hence
\begin{align}
  \rho &= \prod_k \frac{1-i c_k\,z_{2k-1} z_{2k}}2~,
  \label{e.gaussianC}
\end{align}
where $|c_k|\le 1$. Note
that for $c_k=\pm 1$, one has $g_k = \pm\infty$, making the ansatz
\eqref{e.gaussian} not well defined, unlike Eq.~\eqref{e.gaussianC}.
The latter possibility %The singularity  of Eq.~\eqref{e.gaussian}
occurs for instance for pure states, as it is clear from the following
explicit expression for
the purity of the states \eqref{e.gaussian} and the states \eqref{e.gaussian} and  \eqref{e.gaussianC}:
\begin{align}
  \Tr[\rho^2] =
  \frac{\det\left[2\cosh\left(i\,G\right)\right]^{\frac12}}{
  \det\left[2\cosh\left(i\frac{G}2\right)\right]} =
  \sqrt{\det\left(\frac{1+C^2}2\right)}~.
  \label{e.purity}
\end{align}

We now derive the proof of Eqs.~\eqref{e.bures} and \eqref{e.buresbasis}, 
dividing the different
steps into three lemmas. At first we assume $c_k\neq \pm1$
and then we extend the result for including pure states.

\begin{lemma}
  Let $\rho, \rho'$ two GF-states \eqref{e.gaussian}
  parametrized by $G,G'$ respectively.  Then
  \begin{align}
    \mathcal F(\rho, \rho') &= \Tr\sqrt{\sqrt\rho\rho'\sqrt\rho}
    \\&= \frac{\det\left[\openone +
      \sqrt{e^{iG/2}e^{iG'}e^{iG/2}}\right]^{\frac12}}{
	\det\left[\openone + e^{iG}\right]^{\frac14}\,
    \det\left[\openone + e^{iG'}\right]^{\frac14}}~.
    \label{e.uhlmann}
  \end{align}
  \label{l.1}
\end{lemma}
\begin{proof}
  This lemma is a direct consequence of the fact the quadratic Majorana
  operators form a Lie algebra:
  \begin{align}
    \left[\frac{\V w\cdot A \V w}4, \frac{\V w \cdot B\V w}4\right] =
    \frac{\V w\cdot \left[A, B\right]\V w}4~,
  \end{align}
  and accordingly
  \begin{align}
    e^{\frac{i}4 \V w \cdot A \V w}
    e^{\frac{i}4 \V w \cdot B \V w} &=
    e^{\frac{i}4 \V w \cdot D \V w}~,
    & e^A\,e^B=e^D~.
  \end{align}
  Thanks to the above identity
  \begin{align}
    \sqrt{\sqrt\rho\rho'\sqrt\rho} \propto \exp\left(\frac14\sum_{ij}
    \left(\frac{\log\left[e^{-iG/2}e^{-iG'}e^{-iG/2}\right]}2\right)_{ij}\,
    w_iw_j\right)~,
  \end{align}
  and using \eqref{e.gaussianZ} we find
  \begin{align}
    \mathcal F(\rho, \rho') = \frac{\det\left[\cosh\left(\frac14\log
      e^{-iG/2}e^{-iG'}e^{-iG/2}\right)\right]^{\frac12}}{
	\sqrt{\det\left[\cosh\left(i\frac{G}2\right)\right]^{\frac12}
	\,\det\left[\cosh\left(i\frac{G'}2\right)\right]^{\frac12}}}~,
      \end{align}
    which is equivalent to \eqref{e.uhlmann}.
\end{proof}
A convenient parametrization of Eq.~\eqref{e.uhlmann} is obtained in
terms of the correlation function by defining the new matrix
$T=e^{iG}$. Then
\begin{align}
{C = \frac{T-\openone}{T+\openone}}~, && T^T=T^{-1}~, T^\dagger = T~,
\label{e.cayley}
\end{align}
\begin{align}
  \mathcal F(\rho,\rho') =:
  \mathcal F(T, T') = \frac{\det\left[\openone +
  \sqrt{\sqrt T \,T' \sqrt T}\right]^{\frac12}}{
    \det\left[\openone + T\right]^{\frac14}\,
    \det\left[\openone + T'\right]^{\frac14}}~.
    \label{e.fidT}
\end{align}

%Thanks to the above parametrization one can obtain the metric pull back:
The following lemma conveys the metric pull back with in the manifold
of states parametrized by $T$:
\begin{lemma}
Let $ds^2 = 8\,ds^2_B=16[1-\mathcal F(T, T+dT)]$
the fidelity metric around the state \eqref{e.gaussian} pulled back in the
space of the matrices $T$ and let $dT = \partial_\mu T \,d\lambda_\mu$
where $\lambda_\mu\in\mathcal M$ are the parameters
of the model. Then
%where $\partial_\mu \equiv \frac\partial{\partial\eta_\mu}$.
the fidelity metric can be cast
into the form $ds^2 = \sum_{\mu\nu} g_{\mu\nu}\,d\lambda_\mu\,d\lambda_\nu$
where the geometric tensor is
\begin{align}
  g_{\mu\nu} = 2\sum_{ij}\frac{(\partial_\mu T)_{ij}\,(\partial_\nu T)_{ji}}{
  (1+t_i)(1+t_j)(t_i+t_j)} ~.
  \label{e.metricT}
\end{align}
In \eqref{e.metricT}
the sum is performed in the basis in which $T$ is diagonal, i.e. we set
$T=\sum_i t_i \ket i \bra i$ and $(\partial_\mu T)_{ij}=\bra i
\partial_\mu T\ket j$.
  \label{l.2}
\end{lemma}
\begin{proof}
Proceeding along the same lines of Section 3 of \cite{sommers2003bures} we
obtain for $T'=T+dT$
\begin{align}
&\sqrt{\sqrt T T'\!\sqrt{T}} = T + \sum_{ij} \ket i \bra j\,
\frac{\sqrt{t_it_j}}{t_i+t_j}  \,dT_{ij} -\\&-
\sum_{ijk}\ket i\bra k \,dT_{ij}dT_{jk} \,\frac{\sqrt{t_it^2_jt_k}}{
(t_i+t_j)(t_j+t_k)(t_i+t_k)}
+\mathcal O\left(dT\right)^3
\nonumber
\end{align}
Owing to the above expression and to Eq.\eqref{e.fidT} the fidelity
$\mathcal F(T,T+dT)$ can be written
in terms of some infinitesimal operators $\delta, \partial$
\begin{align}
\mathcal F(T, T+dT) &\simeq \frac{\det\left[(\openone + T)(\openone+\partial)
 \right]^{\frac12}}{
    \det\left[\openone + T\right]^{\frac14}\,
    \det\left[(\openone + T)(\openone+\delta)\right]^{\frac14}}
 \nonumber\\&= \frac{\det\left[\openone+\partial \right]^{\frac12}}{
    \det\left[\openone+\delta\right]^{\frac14}}
    = e^{\frac12\Tr \log(1+\partial) -
    \frac14\Tr \log(1+\delta)} \nonumber\\&
    \simeq e^{\frac12\Tr\left(\partial-\delta/2\right) -
    \frac14\Tr\left(\partial^2-\delta^2/2\right)}~,
    \label{sgniafuz}
\end{align}
where
\begin{align}
\delta &= (1+T)^{-1} \,dT = \sum_{ij}\ket i\bra j \,\frac1{1+t_i} \,dT_{ij}~,\\
\partial &= (1+T)^{-1}\,\left(\sqrt{\sqrt T T'\!\sqrt{T}}-T\right)
\\\nonumber&= \sum_{ij} \ket i \bra j \,
\frac{\sqrt{t_it_j}}{t_i+t_j} \,\frac1{1+t_i} \,dT_{ij} -\\\nonumber&-
\sum_{ijk}\ket i\bra k \,dT_{ij}dT_{jk} \,\frac{\sqrt{t_it^2_jt_k}}{
(t_i+t_j)(t_j+t_k)(t_i+t_k)} \,\frac1{1+t_i}~.
\end{align}
The elements of Eq.~\eqref{sgniafuz} become
\begin{align}
\Tr(\partial - \delta/2)
 & = -\frac14
\sum_{ij}\left|dT_{ij}\right|^2 \,\frac1{(t_i+t_j)^2} \,\left(
\frac{t_j}{1+t_i} + \frac{t_i}{1+t_j}\right)~,
\\
\Tr \delta^2 &=
\sum_{ij}\left|dT_{ij}\right|^2 \,\frac1{(1+t_i)(1+t_j)}~,
\\
\Tr \partial^2 &\simeq
\sum_{ij}\left|dT_{ij}\right|^2 \, \frac{t_it_j}{
(t_i+t_j)^2} \,\frac1{(1+t_i)(1+t_j)}~,
\end{align}
so that
\begin{align}
\mathcal F(T, T+dT) \simeq
1- \frac18\sum_{ij}
\frac{\left|dT_{ij}\right|^2}{(1+t_i)(1+t_j)(t_i+t_j)}~,
\label{e.fid2order}
\end{align}
which completes the proof.
\end{proof}

Before proving Eq.~\eqref{e.bures} we introduce the following lemma which
will be used for analytical continuations to the pure state manifold:
\begin{lemma}
Let $f(x,y):= (x-y)^2(1-xy)^{-1}$ be a function defined in
$[-1,\,1]^2-\{ z^+,\,z^-\},\, z^\pm:=(\pm 1,\,\pm 1).$ Then
$f(x,\, y)\le 4$ and  $\lim_{(x,y)\to z^\pm } f(x,\,y)=0$.
  \label{l.3}
\end{lemma}
\begin{proof}
The upper bound is found thanks to $1-xy=1-[(x+y)^2-(x-y)^2]/4\ge(x-y)^2/4$.
In order to show that $\lim_{(x,y)\to z^\pm } f(x,\,y)=0$
{let us restrict $f$ to the $x\ge 0,\, y\ge 0$ part of the domain to analyse the limit to $z^+$.
The limit $z^-$ follows because of the $(x,\, y)\to (-x,\,-y)$ symmetry of $f.$
 One can write $y=1+ m\,(x-1)$  or $x=1+ m\,(y-1)$ with with $m\in[0,\,1].$ Because of the $(x,\,y)\to (y,\,x)$ symmetry of $f$ we can consider just the first case.
One obtains $f(x,y)= (1-x)\frac{(1-m)^2}{1+m x}\le 1-x$ this quantity in a disk of radius $\delta$ centered on $z^+$ is upper bounded by $\delta$.
This shows that $\forall\epsilon>0,\,\exists \delta=\delta(\epsilon)$ s.t $\|(x,\,y)-z^+\|\le \delta\Rightarrow f(x,\,y)\le \epsilon$ (with $\delta(\epsilon)=\epsilon$), i.e., the claim. }
\end{proof}

%We are finally ready for
\begin{proof}[Proof of Eq.~\eqref{e.bures}]
  Eq.~\eqref{e.buresbasis} is obtained directly from lemma~\ref{l.2}.
  Indeed, from Eq.~\eqref{e.cayley}
\begin{align}
dC=dT\frac1{1+T}-\frac{T-1}{T+1}\,dT\,\frac1{T+1} = 2\frac1{T+1}dT\frac1{T+1}~.
\end{align}
Inserting the above equation in \eqref{e.metricT}, and noting that
$C$ and $T$ are diagonal in the same basis,
$c_i = \frac{t_i-1}{t_i+1}$, one obtains
\begin{align}
g_{\mu\nu} = \sum_{ij}\frac{(\partial_\mu C)_{ij}\,
(\partial_\nu C)_{ji}}{1-c_ic_j}~.
\label{e.metricC}
\end{align}

The singular behaviour of
\eqref{e.metricC} for  $c_i = \pm1$ is just apparent. Indeed,
let $iG|j\rangle=g_j |j\rangle \,(j=1,\ldots,2n),\,{\rm{Sp}}(iG)=\{g_j\}\subset{\mathbf{R}}$ then $C=\sum_j c_j
|j\rangle\langle j|,\,c_j:=\tanh(g_j/2).$
%[Since $c_j{\to}\pm 1\Leftrightarrow \lambda_j{\to}\pm\infty$ no $C$ s.t. $1\notin{\rm{Sp}}(|C|)$
%is in the image of any $G.$]
By differentiation
$
dC=\sum_j\left(  (1-c_j^2) \frac{dg_j}{2} |j\rangle\langle j| + c_j (|dj\rangle\langle j|+|j\rangle\langle dj|)\right).
$
One has therefore the following matrix elements $(dC)_{jj}= (1-c_j^2) dg_j $ and $(dC)_{ij}=(c_i-c_j) \langle di|j\rangle,\,(i\neq j).$
Plugging these in \eqref{e.metricC}
\begin{equation}
ds^2=\frac{1}{4} \sum_j (1-c_j^2)\, dg_j^2 +\sum_{i\neq j} %\frac{(c_i-c_j)^2}{1-c_i c_j}
f(c_i, c_j) \, |\langle di|j\rangle|^2~.
\label{apparent}
\end{equation}
Now one sees easily that for $c_j\to \pm 1$ the first (diagonal) contribution in (\ref{apparent}) vanishes while the second, thanks to lemma~\ref{l.3},
is  upper bounded by $4 \,\sum_{i\neq j} |\langle di|j\rangle|^2$ for {\em all} $c_i, c_j\in(-1,\,1)$ and vanishes for $(c_i,c_j)\to z^{\pm}$:
%This result about the lack of singularities for $c_i c_j=1$ naturally suggest that,
even if \eqref{e.metricC} has been derived for $C$ such that $c_i\neq \pm1$,
we can perform the limit $|c_i|\to 1, (\forall i)$ and, in this way,
{\em extend} the metric to the pure state manifold
just by setting $c_i c_j$ to $-1$ (as for the case $c_i c_j=1$ gives vanishing contribution).
%If this way one finds $g_{\mu \nu}= \frac{1}{2}
%\Tr\left[ \partial_\mu C\,\partial_\nu C\right]$.

The basis independent expression Eq.~\eqref{e.bures} follows from
\eqref{e.metricC} 
%From this one gets Eq.~\eqref{e.bures} 
\begin{equation}
ds^2=\sum_{\mu\nu} g_{\mu\nu} d\lambda_\mu d\lambda_\nu=\langle ({\mathbf{1}}-{\rm{Ad}} C)^{-1}(dC),\, dC\rangle
\label{ds^2}
\end{equation}
where $dC=\sum_\mu d\lambda_\mu \partial_\mu C,$ and $ {\rm{Ad}} C( X):= C X C^\dagger=CXC=(L_C\circ R_C)(X)$ is the adjoint
action. 
To see this let us
first write $dC=\sum_{ij} (dC)_{ij} |i\rangle\langle j|$ where
$C|i\rangle=c_i |i\rangle.$ Then $(1-{\rm{Ad}} C)^{-1}(dC)  =\sum_{ij} (dC)_{ij}  (1-c_i c_j )^{-1} \,|i\rangle\langle j|$
and $\langle (1-{\rm{Ad}} C)^{-1}(dC),\,dC\rangle=  \sum_{ij} {(dC)_{ij}}^*  (1-c_i c_j )^{-1} \langle | i\rangle\langle j|,\,dC\rangle=
\sum_{ij} {(dC)_{ij}}^*  (dC)_{ij} (1-c_i c_j )^{-1}.$
%and from the definition of ${\rm Ad}_C$; 
The zero contribution to the sum \eqref{e.metricC} for $c_ic_j=1$ is considered
thanks to the pseudo-inverse.
%The sum over the elements such that $c_ic_j\neq 1$ can be written in the
%base-independent expression \eqref{ds^2} and \eqref{e.bures}
%provided that the symbol $^{-1}$ refers to the pseudo-inverse.
%Indeed, the operator $\hat K = 1+C^{\otimes 2}$ is
%non-negative and diagonalizable: $\hat K = \hat W\, \hat k\,\hat{W}^\dagger$,
%with a diagonal matrix $\hat k$ whose diagonal entries are given by
%$\hat k_{\hat \ell}=1+c_\ell c_{\ell'}$,  with a proper multi-index $\hat \ell = (\ell,\ell')$.
%The pseudo-inverse is thus given by $\hat{K}^{-1} =
%\hat{W}^\dagger\,\hat{k}^{-1}\,\hat W$ where the $\hat \ell$-th entry of
%the diagonal matrix $\hat{k}^{-1}$ is $\hat k_{\hat \ell}^{-1}$ if
%$\hat k_{\hat \ell} = 1+c_\ell c_{\ell'}\neq 0$ while it is zero elsewhere.
\end{proof}

One can show that Eq.~\eqref{e.bures} reduces to
the known expressions when $\rho$ is a thermal state
\cite{zanardi2007bures} and when $\rho$ is a pure state
\cite{cozzini2007quantum}, provided that the appropriate matrices
$T$ or $C$ are used. In the next section, this theorem is applied
to NESS-QPT where $C$ is given by the solution of the Sylvester
equation \eqref{e.C}.

\section{Liouvillean steady state}
We call $\mathcal R$ the $4^n$-dimensional operator spaces generated by
$\prod_j w_j^{s_j}$, ($s_j\in\{0,1\}$), and we use the
notation $\sket{\V s}$ for referring to the elements of $\mathcal R$,
normalized with respect to the Hilbert-Schmidt inner product, i.e.
$\sbraket{\V s}{\V s} \equiv \Tr[s^\dagger s]= 1$ for
$\sket{\V s}\in\mathcal R$. 

Following the notation introduced in the Letter, the Liouvillean
$\mathcal L\colon \mathcal R\rightarrow \mathcal R$ introduced in 
\eqref{e.lind} can be written as
\begin{align}
  \mathcal L = -\frac12\begin{pmatrix} {\V a}^\dagger & \V a \end{pmatrix}
  \,
  \begin{pmatrix}
    X & Y \\ 0 & -X^T
  \end{pmatrix}
  \,
  \begin{pmatrix} \V a \\ {\V a}^\dagger \end{pmatrix}  -\frac12 \Tr X~.
\end{align}
The superoperator $a_j^\dagger$ is the Hermitian conjugate of $a_j$ in $\mathcal R$.

If $C$ is the matrix solution of \eqref{e.C} then
 \begin{align}
   \begin{pmatrix}
    X & Y \\ 0 & -X^T
  \end{pmatrix}
  =
  \begin{pmatrix}
    U & -C\,U^{-T}\\0&U^{-T}
  \end{pmatrix}
  \,
   \begin{pmatrix}
     x & 0 \\ 0 & -x
  \end{pmatrix}\,
  \begin{pmatrix}
    U^{-1} & U^{-1}\,C\\0&U^T
  \end{pmatrix}~.
  \label{e.diaglind}
 \end{align}
We show now that the latter transformation is non-unitary Bogoliubov
transformation \cite{blaizot1985quantum} and that everything is consistent.
It is known that non-unitary Bogoliubov transformations are isomorphic to
the group of orthogonal complex matrices $O(4n, \mathbb C)$. This condition
can be expressed in a simple way thanks to
Eq.(2.6) of \cite{blaizot1985quantum}, i.e.
\begin{align}
  \hat V\,\Sigma^x\,\hat V^T &= \Sigma^x~, & \Sigma^x &= \sigma^x \otimes
  \openone_{2n}~.
\end{align}
It is simple to show that the transformation $\hat V$
\begin{align}
      \hat V &=\begin{pmatrix}
	U^{-1} & U^{-1}\,C\\0&U^T~,
      \end{pmatrix}
\label{e.bogohat}
\end{align}
satisfies that condition. 
We define 
new diagonal creation and annihilation operators as
\begin{align}
  \begin{pmatrix} \V d \\ {\V d}^\times \end{pmatrix} &= \hat V \,
    \begin{pmatrix} \V a \\ {\V a}^\dagger \end{pmatrix}~.
      \label{e.bogodiag}
\end{align}
Since $\mathcal V$ is a non-unitary Bogoliubov transformation the operators 
$d_i$ and $d_j^\times$ satisfy the CAR-algebra, but 
$d_j^\times \neq d_j^\dagger$.
Moreover, using
$ \begin{pmatrix} {\V a}^\dagger & \V a \end{pmatrix} =
  \begin{pmatrix} \V a \\ {\V a}^\dagger \end{pmatrix}^T\,\Sigma^x$ then
it is simple to show that
\begin{align}
  \mathcal L = -\frac12\begin{pmatrix} {\V d}^\times & \V d \end{pmatrix}
  \,
  \begin{pmatrix}
    x & 0 \\ 0 & -x
  \end{pmatrix}
  \,
  \begin{pmatrix} \V d \\ {\V d}^\times \end{pmatrix}  -\frac12 \Tr X~,
\end{align}
i.e.,
\begin{align}
  {\mathcal L = -\sum_j x_j \,d^\times_j\, d_j}~.
  \label{e.lindiag}
\end{align}
Note also that the transformation \eqref{e.bogodiag} can be written
thanks to Eq (2.16) of \cite{blaizot1985quantum} into the form
\begin{align}
  d_j &= \mathcal V\, a_j\, \mathcal V^{-1}~, &
  d_j^\times &= \mathcal V\, a_j^\dagger\, \mathcal V^{-1}~,
\end{align}
where
\begin{align}
  \mathcal V = :\exp\left(-\frac12 {\V a}^\dagger \,C\, {\V a}^\dagger +
  {\V a}^\dagger\,(U-1)\,{\V a}\right):~,
\end{align}
and $:\exp(\cdot):$ refers to the normal ordering of the exponential.

It is now possible to express 
the stationary state of the Liouvillean, i.e. the
state $\Omega$ such that $\mathcal L \Omega = 0$, as the ${\V d}$-vacuum,
%of the diagonal operators,
i.e. $d_j \sket\Omega=0$. 
The identity operator, i.e. the element
$\sket{\V 0} \in \mathcal R$ is the $\V a$-vacuum, i.e.
$a_i \sket{\V 0} = 0$, $\forall j=1,\ldots2n$, and in particular
$\sbra{\V 0}\mathcal L = 0$. The $\V d$-vacuum can be readily obtained
from the Bogoliubov transformation: 
$\sket\Omega=\mathcal V\sket{\V 0}$. 
Indeed, as $a_j \sket{\V 0}=0$, one has $ d_j\sket\Omega =
\mathcal V a_j\mathcal V^{-1}\mathcal V\sket{\V 0} = 0$.
Hence,
\begin{align}
  \sket{\Omega} =
  \mathcal V \sket{\V 0} = e^{-\frac12 {\V a}^\dagger\,C\,{\V a}^\dagger}
  \sket{\V 0}~.
  \label{e.superss}
\end{align}
We now show that the state \eqref{e.superss} is exactly \eqref{e.gaussianC}.
Thanks to the transformation $Q$ defined in \eqref{e.Q} and the direct
relation \eqref{e.corrG} one can write the imaginary antisymmetric matrix
$C=Q^T\,\bigoplus_k \begin{pmatrix}
  0&ic_k\\-ic_k&0
\end{pmatrix}\,Q$. Then, using the definition \eqref{e.superss}
\begin{align}
  \frac12 {\V a}^\dagger \, C \, {\V a}^\dagger \rho &=
   \frac18\left(\V w\cdot C\V w\rho + 2 \V w\cdot C\rho \V w +
   \rho \V w \cdot C \V w\right) \nonumber\\&
   = \frac{i}4 \sum_k c_k \big[z_{2k-1}z_{2k}\rho
   +z_{2k-1}\rho z_{2k}-\nonumber\\&\quad\quad\quad\quad-z_{2k}\rho z_{2k-1}+
   \rho z_{2k-1}z_{2k}\big] \nonumber\\&=: \sum_k \mathcal G_k(\rho)~.
\end{align}
As
\begin{align}
  \mathcal G_k(\openone) &= i \,c_k \,z_{2k-1}z_{2k}~, &
  \mathcal G_k(z_{2k-1}z_{2k}) = 0~,
\end{align}
it is clear that
\begin{align}
  \Omega \propto e^{-\frac12 c^\dagger C c^\dagger} \sket{0} \propto
  \prod_{k} e^{-\mathcal G_k} \openone =
  \prod_k \left(1 - i \,c_k \,z_{2k-1}z_{2k} \right)~,
  \label{e.rhoz}
\end{align}
thus recovering Eq.\eqref{e.gaussian}.

The conditions for the existence and uniqueness of \eqref{e.rhoz}
are given in \cite{prosen2010spectral}. We now study those conditions and
express them in terms of the spectral gap.
The correlation matrix matrix $C\in M_{2n}(\mathbf{C})$ is the matrix solution of Eq.~\eqref{e.C}.
To study the solution of that equation it  is useful to consider the (non-canonical) ``vectorising"  isomorphism $\phi\colon M_{2n}({\mathbf{C}})\rightarrow  ({\mathbf{C}}^{2n})^{\otimes\,2}\,/\, |i\rangle\langle j|\rightarrow |i\rangle\otimes|j\rangle.$ This is also a Hilbert-space
isomorphism, namely $\langle \phi(A),\phi(B)\rangle = \langle A,\,B\rangle={\rm{Tr}}\,(A^\dagger B).$
One can directly check that if $R_X(C):=CX$ and $L_X(C):=XC$ then
$\phi(R_X(C))=(\phi\circ R_X\circ \phi^{-1} \circ\phi)(C)= ({\mathbf{1}}\otimes X^T)\phi(C),$ and $
\phi(L_X(C))=(\phi\circ L_X\circ \phi^{-1} \circ\phi)(C)= (X\otimes {\mathbf{1}})\phi(C).
$ Applying $\phi$ to both sides of \eqref{e.C} one then obtains ($\tilde C :=\phi(C),\,\tilde Y:=\phi(Y)$)
%\end{equation}
\begin{equation}
(X\otimes {\mathbf{1}} +{\mathbf{1}}\otimes X) \tilde{C}=: \hat{X} \tilde{C}=\tilde{Y},
\label{syl1}
\end{equation}
where $\tilde{C},\tilde{Y}\in({\mathbf{C}}^{2n})^{\otimes\,2},\,\hat{X}\in{\rm{End}}({\mathbf{C}}^{2n})^{\otimes\,2}\cong M_{4n^2}({\mathbf{C}}).$
There are three different key operators
in the formalism for obtaining the steady state:
\begin{enumerate}
\item The Liouvillean ${\cal L}\colon {\rm{End}}(  ({\mathbf{C}}^2)^{\otimes n} )\rightarrow{\rm{End}}(  ({\mathbf{C}}^2)^{\otimes n} ),$
 a $2^{2n}\times 2^{2n}$ matrix. Its complex spectrum, from \eqref{e.lindiag}, is given by
\begin{align}
{\rm{Sp}}({\cal L})= -\{x_{\mathbf{n}}:=\sum_{j=1}^{2n} x_j n_j\,/\, n_j=0,1,\,x_j\in {\rm{Sp}}(X)\}.
\label{e.spectrum}
\end{align}
Notice that $0\in {\rm{Sp}}({\cal L})$ i.e., $\cal L$ is always non-invertible and that the steady state(e.g., our Gaussian one ${\mathbf{n}}={\mathbf{0}}$) are in the kernel of $\cal L$.
If this latter is one-dimensional (unique steady state) the gap of  $\cal L$ can be defined as
$
{\Delta}_{\cal L}:=\min_{  {\mathbf{n}} \neq {\mathbf{0}}   }\, |x_{\mathbf{n}}|.
$

\item The map $X\colon {\mathbf{C}}^{2n}\rightarrow  {\mathbf{C}}^{2n},$ a $2n\times 2n$ {\em real} diagonalizable matrix. Its spectrum
is $\{x_j\}_{j=1}^{2n}\subset{\mathbf{C}}$ and (because of reality) is invariant under complex conjugation. On physical grounds (stability) we {\em must}
have $\Re\,x_j\ge0,\forall j.$
Indeed, the time-scale for convergence $\rho(t)\to\rho(\infty)$ is dictated by $\tilde{\Delta}^{-1}$ where $\tilde\Delta= \min_{{\mathbf{n}} \neq {\mathbf{0}}} \Re\, x_{{\mathbf{n}}}.$

\item The map $\hat{X}=X\otimes {\mathbf{1}} +{\mathbf{1}}\otimes X\colon {\mathbf{C}}^{2n}\otimes {\mathbf{C}}^{2n}\rightarrow {\mathbf{C}}^{2n}\otimes {\mathbf{C}}^{2n},$
 a $4n^2\times 4n^2$ matrix. It spectrum is $\{ x_i+x_j\}_{i,j=1}^{2n}\subset{\mathbf{C}}$ and the minimum (in modulus) is given by $\Delta_{\hat{X}}:=\min_{i,j} |x_i+x_j|.$
Note also that 
\begin{align}
\Delta_{\hat X}^{-1} = \|\hat{X}^{-1}\|_\infty~.
\label{e.deltaxhat}
\end{align}
%From this we see that \eqref{e.C} has a unique solution iff the linear map $\hat{X}$ over $({\mathbf{C}}^n)^{\otimes\,2}$  is invertible i.e., $0\notin {\rm{sp}}\, (\hat{X}).$
For the  uniqueness of the steady state we must have $\hat{X}$ invertible i.e., $\Delta_{\hat{X}}>0.$
\end{enumerate}

\begin{prop} If $\Delta = \min_j 2\Re(x_j) > 0 $ then
\begin{equation}
 \Delta={\Delta}_{\cal L}= \Delta_{\hat{X}}~.
 \label{e.prop1}
\end{equation}
\end{prop}
\begin{proof}
$|x_{\mathbf{n}}|=|\sum_{j=1}^{2n} x_j n_j|\ge |\Re(\sum_{j=1}^{2n} n_j  x_j) |.$ The first bound can be saturated by choosing the $n_j$'s in such a way that only a set $P$ of complex conjugated pairs
$x^\pm_p$
of eigenvalues are present. In this case $ |\Re(\sum_{j=1}^{2n} n_j  x_j|=2\sum_{p\in P} \Re\,x_p.$ Where we used  the  assumption $(\forall p)\, \Re\,x_p\ge 0.$
Using again positivity of all the terms, this sum can be made as small as possible by choosing $|P|=1$ and minimizing over $p=1,\ldots,n.$
This shows that ${\Delta}_{\cal L}=\min_{\mathbf{n}} |x_{\mathbf{n}}|=2 \min \{\Re\,x_p\}_{p=1}^n.$ It is clear now that a  similar argument  shows that $\Delta_{\hat{X}}=\min \{|x_i+x_j|\}_{i,j=1}^{2n}$
is given by the same expression i.e. ${\Delta}_{\cal L}= \Delta_{\hat{X}}$. Finally $\Delta=2\min_{\mathbf{n}} \Re\, x_{\mathbf{n}}\equiv2\tilde\Delta=2\min_p \Re\, x_p=\Delta_{\cal L}$.
\end{proof}

\section{Non-diagonalizable case}
The non-diagonalizable case has been extensively handled in 
\cite{prosen2010spectral}. In the previous section we have assumed $X$ to be
diagonalizable for simplicity, and because the matrices $X$ 
encountered in our numerical simulations 
were diagonalizable. Here we briefely discuss the
general case. The matrix $X$ can always be put in the Jordan canonical form, 
i.e. $X=U\,x^J\,U^{-1}$ with $x^J = \oplus_b J_{\ell_b}(x_b)$, 
\begin{align}
J_{\ell_b}(x_b) = \begin{pmatrix} 
x_b & 1 & \\
& x_b & 1 & \\
&& x_b & 1 & \\
&&& \ddots & \ddots &
\end{pmatrix}~:
\end{align}
$x_b$ are (possibly equal) eigenvalues of $X$ and $\ell_b$ 
is the dimension of the Jordan block: each block is composed of $\ell_b$ 
degenerate eigenvalues of $X$.  
The form of the transformation \eqref{e.bogohat} remains the same 
(although with a new matrix $U$) while \eqref{e.lindiag} becomes 
\begin{align}
\mathcal L = 
-\sum_{j=1}^{2n} x_j \, d_j^\times d_j - \sum_b\sum_{k=1}^{\ell_b-1}
d^\times_{b_k+1} d_{b_k}~,
\label{e.linjor}
\end{align}
where $b_k$ refers to the index of the $k$th element in the $b$th Jordan 
block. It is clear that the state \eqref{e.superss} is still a stationary
state. Moreover, in \cite{prosen2010spectral} it has been shown that the
spectrum of the Liuvillean is 
\begin{align}
{\rm{Sp}}({\cal L})= 
-\{x_{\mathbf{n}}:=\sum_{b} x_b n_b\,/\, n_b=0,\cdots,\ell_b\}.
\label{e.spectrumnd}
\end{align}
Accordingly, $\Delta_{\mathcal L} = \Delta \equiv 2\min_b \Re[x_b]$.
If $\Delta>0$ the steady state \eqref{e.superss} is unique 
\cite{prosen2010spectral}.

In the non-diagonalizable case the last equation in Eq.~\eqref{e.prop1} 
is not satisfied. On the other hand one can obtain the following
\begin{prop}\begin{equation}
\|\hat{X}^{-1}\|_\infty < \frac{1+p(\Delta^{-1})}\Delta~,
\label{e.deltadeltand}
\end{equation} 
for a certain polynomial $p()$.
\end{prop}
\begin{proof}
We start by writing 
\begin{align}
\hat X &= 
\bigoplus_b J_{\ell_b}(x_b)\otimes \openone + \bigoplus_b \openone
\otimes J_{\ell_b}(x_b)
\nonumber\\&=
\bigoplus_{b,d} \left[J_{\ell_b}(x_b)\otimes \openone_{\ell_d} + 
\openone_{\ell_b} \otimes J_{\ell_d}(x_d)\right]
\nonumber\\&=\hat x + 
\bigoplus_{b,d} \left[J_{\ell_b}(0)\otimes \openone_{\ell_d} + 
\openone_{\ell_b} \otimes J_{\ell_d}(0)\right]~,
\end{align}
where $\hat x$ is the diagonal matrix with entries $x_i + x_j$ and where
we used the decomposition $\openone = \oplus_b 1_{\ell_b}$. Moreover,
thanks to Lemma 3.1 of Ref.~\cite{prosen2010spectral}, 
\begin{align}
\hat X &= 
\hat x + 
\bigoplus_{b,d} \bigoplus_{r=1}^{\min\{\ell_b,\ell_d\}} 
J_{\ell_b+\ell_d -2r+1}(0)
\nonumber\\&=
\hat x\left[\openone + 
\bigoplus_{b,d} \bigoplus_{r=1}^{\min\{\ell_b,\ell_d\}} 
\frac{J_{\ell_b+\ell_d -2r+1}(0)}{x_b+x_d}\right]~.
\end{align}
As $J$ is nilpotent,
\begin{align*}
\hat{X}^{-1} &= 
\hat{x}^{-1}
\left[\openone + 
\bigoplus_{b,d} \bigoplus_{r=1}^{\min\{\ell_b,\ell_d\}} 
\sum_{m=1}^{\ell_b+\ell_d-2r}\left(-
\frac{J_{\ell_b+\ell_d -2r+1}(0)}{x_b+x_d}\right)^m\right]~,
\end{align*}
and
\begin{align}
\|\hat{X}^{-1}\|_\infty &\le 
\|\hat{x}^{-1}\|_\infty
\left[1 + \max_{b,d} \max_r %\max_{r=1}^{\min\{\ell_b,\ell_d\}} 
\sum_{m=1}^{\ell_b+\ell_d-2r} \frac{1}{|x_b+x_d|^m}\right]
\nonumber\\&=
\|\hat{x}^{-1}\|_\infty
\left[1 + \max_{b,d} \sum_{m=1}^{\ell_b+\ell_d-2} \frac{1}{|x_b+x_d|^m}\right]
\nonumber\\&\le\frac1\Delta
\left[1 + \max_{b,d} \sum_{m=1}^{\ell_b+\ell_d-2} \frac{1}{\Delta^m}\right]~.
\end{align}
\end{proof}

\section{Upper bounds}
In order to derive some bounds to the fidelity metric $ds^2$ let us
express Eq.~\eqref{e.bures} in a convenient form thanks to
the vectorization isomorphism. As
 $ {\rm{Ad}}_C( X)=(L_C\circ R_C)(X)$
one has $ \phi\circ(L_C\circ R_C)\circ \phi^{-1}=C\otimes C^T=-C^{\otimes\,2}$
and Eq.~\eqref{e.bures} becomes
\begin{equation}
ds^2 =\langle ({\mathbf{1}} +C^{\otimes\,2})^{-1}(d\tilde{C}),\, d\tilde{C}\rangle
%\langle ({\mathbf{1}} +C^{\otimes\,2})^{-1/2}(d\tilde{C}),\, ({\mathbf{1}} +C^{\otimes\,2})^{-1/2}( d\tilde{C})\rangle
=\| ({\mathbf{1}} +C^{\otimes\,2})^{-1/2}(d\tilde{C})\|^2~,
\label{ds^2_1}
\end{equation}
where $d\tilde C=\phi(dC).$
%{\em By plugging Eq (\ref{diff}) into (\ref{ds^2_1}) one obtains the pull-back metric on the control parameter $(X,Y)$-manifold.}
Using the Cauchy-Schwarz inequality and the definition of operator norm one obtains
\begin{align}
  \nonumber
 ds^2&\le \|({\mathbf{1}} +C^{\otimes\,2})^{-1}(d\tilde C)\| \|d\tilde C\| \le %\|({\mathbf{1}} +C^{\otimes\,2})^{-1}\|_\infty
 P_C\,\|d\tilde{C}\|^2  \\ &\le2n P_C\,\|d{C}\|^2_\infty~,
  \label{e.cs}
\end{align}
where we have exploited the fact that, by construction,
$\|\tilde{A}\|:=\|\phi(A)\|=\|A\|_2$ and $\|A\|_2\le\sqrt{2n} \|A\|_\infty.$
Now ${\rm{Sp}} (C^{\otimes\,2})=\{ c_i c_j\,/\, c_i, c_j\in{\rm{Sp}}(C)\}$ and, from $C=-C^T,$ the spectrum of $C$ is invariant under $c_i\rightarrow -c_j,$ it follows that $ \|({\mathbf{1}} +C^{\otimes\,2})^{-1}\|_\infty =(1+\min_{i,j} c_i c_j)^{-1}=(1-\max_i c_i^2)^{-1} = (1-\|C\|^2_\infty)^{-1}$.
The bound \eqref{e.cs} is not specific to dissipative quadratic Liouvillean.
In order to connect Eq.\eqref{e.cs} with the properties of the
Liouvillean \eqref{e.lindiag} we differentiate Eq.~\eqref{syl1}
\begin{equation}
d\tilde{C}= \hat{X}^{-1} d\tilde{Y}- \hat{X}^{-1} d\hat{X}  \tilde C~.
%\hat{X}^{-1} \tilde{Y}
\label{e.Cdiff}
\end{equation}
As $d\equiv \sum_\mu d\lambda_\mu\partial_\mu$ the above equation can
be conveniently calculated via
\begin{align}\label{e.diffcorness}
X\, \left(\partial_\mu C\right) +
\left(\partial_\mu C\right)\,X^T = \partial_\mu Y
-\left(\partial_\mu X\right)\, C - C\,
\left(\partial_\mu X^T\right)~,
\end{align}
i.e. the matrices $\partial_\mu C$ entering in \eqref{e.metricC} can be obtained by
solving a new Sylvester equation where the matrices $X,Y,\partial_\mu X,
\partial_\mu Y$ are given by the model.
Taking norms in $({\mathbf{C}}^{2n})^{\otimes\,2}$
%\footnote{
%$\|O\|_\infty:= \sup_{\|v\|=1}\|Ov\|=$largest singular value of $O;$ Notice that $\|O v\|\le\|O\|_\infty \|v\|.$
%$\|O\|^2_2={{\rm{Tr}}(O^\dagger O) }=$sum of the squares of the singular values of $O.$
%}
\begin{align}
\|d\tilde{C}\|  &\le  \|\hat{X}^{-1}\|_\infty(\|d\tilde{Y}\| + \|d\hat{X}\|_\infty \|\tilde{C}\|)\nonumber\\ &=
 \|\hat{X}^{-1}\|_\infty(\|d{Y}\|_2 + \|d\hat{X}\|_\infty \|{C}\|_2)\nonumber \\ & \le
\sqrt{2n} \|\hat{X}^{-1}\|_\infty(\|d{Y}\|_\infty + \|d\hat{X}\|_\infty \|{C}\|_\infty)
\nonumber\\&\le
\sqrt{2n} \|\hat{X}^{-1}\|_\infty(\|d{Y}\|_\infty + \|d\hat{X}\|_\infty)~,
\end{align}
where, among other things, %we have exploited the fact that, by construction, $\|\tilde{A}\|:=\|\phi(A)\|=\|A\|_2$ and $\|A\|_2\le\sqrt{n} \|A\|_\infty.$
%Moreover,
we used the inequality $\|C\|_\infty\le1$ which follows
%for consistency
from the anstaz \eqref{e.corrG}.
 %or, alternatively, from the purity \eqref{e.purity}.
In summary we have the following upper bound on the squared Hibert-Schmidt norm of $dC$ in terms of the
control parameters and their differentials i.e., $X,\,dX$ and $Y,\,dY$
\begin{equation}
\|d\tilde{C}\|^2\le 2 n \|\hat{X}^{-1}\|_\infty^2 (\|d{Y}\|_\infty + 2\|d{X}\|_\infty)^2
\label{e.dcbound}
\end{equation}
where we also used $\|d\hat{X}\|_\infty=\| dX\otimes{\mathbf{1}} +{\mathbf{1}}\otimes dX\|_\infty\le 2\|dX\|_\infty$.
%The quantity $\|d\hat{X}\|_\infty$ can be written in term of the gap as
%\begin{align}
%  \|\hat X^{-1}\|_\infty &
%  \le \|U^{-1}\otimes U^{-1}\|_\infty \, \|(x\otimes\openone+\openone\otimes x)^{-1}\|_\infty \, \|U\otimes U\|_\infty
%  \nonumber\\ &
%  \le \left(\min_{ij} |x_i+x_j|\right)^{-1} \, \frac{\max_{ij} u_iu_j}{\min_{ij} u_i u_j}
%  \nonumber\\&
%  \equiv \kappa^2(U) / \Delta~,
%  \label{e.addendum}
%\end{align}
%where $u_i$ are the singular values of $U$ and $\kappa(U)$ is the condition number of $U$.
%If $U$ is unitary then $u_i \equiv 1$ and $\kappa(U) = 1$.
%Owing to Eqs. \eqref{e.dcbound} and \eqref{e.addendum},
Pluggin the above equation in \eqref{e.cs} and using Proposition 1 
one then obtains
%Using Chauchy-Schwarz and the definition of operator norm
 %$ds^2\le \|({\mathbf{1}} +C^{\otimes\,2})^{-1}(d\tilde C)\| \|d\tilde C\| \le \|({\mathbf{1}} +C^{\otimes\,2})^{-1}\|_\infty \|d\tilde{C}\|^2$ one obtains \eqref{e.corrbound} and
the bound \eqref{e.dsbound}.

Note that in the non-diagonalizable case there is a correction to 
Eq.~\eqref{e.dsbound} due to the polynomial $p$ in \eqref{e.deltadeltand}. 
However, this correction does not alter the main conclusion of bound
\eqref{e.dsbound}: a superextensive behaviour of $ds^2$ implies the closing
of the Liuvillean gap.

\section{Application II: translationally invariant case}
In this section we study a simpler model where all the informations about
the phase transition can be obtained analytically. The model consists of
a fermionic chain on a ring described the Hamiltonian
\begin{align}
  H = \sum_{i=i}^n \left(f_i^\dagger f_{i+1} + 
  \gamma f_i^\dagger f_{i+1}^\dagger + h f_i^\dagger f_i\right) + 
  \text{h.c.}~.
  \label{e.fermichain}
\end{align}
Owing to the Jordan-Wigner transformation, the above model can be mapped into
the XY spin model \eqref{e.xy}, though with closed boundary conditions. 
The interaction with the environment is described by the following Lindblad
operators $L_i^- = \epsilon \mu f_i$, $L_i^+=\epsilon \nu f_i^\dagger$:
they describe the competition between particle-loss and particle-gain processes.
The quadratic Liouvillean is translationally invariant and can be diagonalized
with a Fourier transformation together with a Bogoliubov transformation.
In the Fourier basis, the two point correlation function matrix 
takes the following form \cite{horstmann2013noise} in the weak coupling
limit $\epsilon\to 0$
\begin{align}
  C & %= i\Lambda \bigoplus_n \Re[h_n]\begin{pmatrix}
    %0 &1/h_n^* \\ -1/h_n & 0 \end{pmatrix}
    = i\frac{\Lambda}{2} \bigoplus_k \begin{pmatrix}
    0 &1+e^{iq_k} \\ -1-e^{-iq_k} & 0 \end{pmatrix}
    %& e^{iq_k} = \frac{h_n}{h_n^*}
    ~,
\end{align}
where
\begin{align}
  q_k = -2\arctan\left(\frac{\gamma\sin \phi_k}{h-\cos \phi_k}\right)~,
  \label{e.qn}
\end{align}
being $\phi_k = 2\pi k/n$, $n$ the length of the chain, and 
$\Lambda = \frac{\nu^2-\mu^2}{\nu^2+\mu^2}$.

The above matrix can be diagonalized via the following transformation
\begin{align}
  C %&= \frac{\Lambda}2 \bigoplus_n \frac1{\sqrt 2}\begin{pmatrix}
    %i\frac{1+e^{i q_k}}{|1+e^{i q_k}|} & i\frac{1+e^{i q_k}}{|1+e^{i q_k}|} \\
    %1 & -1
  %\end{pmatrix} \begin{pmatrix}
  %  |1+e^{i q_k}| & 0 \\ 0 & -|1+e^{i q_k}| 
  %\end{pmatrix}
  %\frac1{\sqrt 2}
  %\begin{pmatrix}
  %  -i\frac{1+e^{-i q_k}}{|1+e^{i q_k}|} & 1 \\
  %  -i\frac{1+e^{-i q_k}}{|1+e^{i q_k}|} & -1
  %\end{pmatrix}
  %\\ &
    = \Lambda \bigoplus_k &\frac1{\sqrt 2}\begin{pmatrix}
    i e^{iq_k/2} & ie^{i q_k/2} \\
    1 & -1
  \end{pmatrix} \begin{pmatrix}
    \cos \frac{q_k}2 & 0 \\ 0 & -\cos \frac{q_k}2
  \end{pmatrix}\times\nonumber\\&\times
  \frac1{\sqrt 2}
  \begin{pmatrix}
    -ie^{-i q_k/2} & 1 \\
    -ie^{-i q_k/2} & -1
  \end{pmatrix}~.
\end{align}
Hence $\|C\|_\infty \simeq |\Lambda|$ and therefore, for consistency, one
has to assume $|\Lambda|\le 1$. Similarly,
%\begin{align}
%  dC &= -\frac{\Lambda}{2} \bigoplus_n \begin{pmatrix}
%    0 &e^{iq_k} \\ e^{-iq_k} & 0 \end{pmatrix} d q_k
%\end{align}
%and 
in the basis in which $C$ is diagonal,
\begin{align}
  dC &= \frac{\Lambda}{2} \bigoplus_k \begin{pmatrix}
    \sin \frac{q_k}2 &i\cos \frac {q_k} 2 \\ -i\cos \frac {q_k} 2 & -\sin \frac {q_k} 2 
  \end{pmatrix} d q_k
\end{align}
so that
\begin{align}
  ds^2 = \frac{\Lambda^2}{2}\sum_k 
  \frac{1-\Lambda^2\cos^2\frac {q_k} 2 \cos q_k}{1-\Lambda^4\cos^4 \frac {q_k} 2}
  (dq_k)^2~.
  \label{e.ds2cirac}
\end{align}
Moreover,
\begin{align}
  d q_k = 2\gamma\frac{\sin\phi_k}{\omega_k^2}~dh -2
  \frac{(h-\cos \phi_k)\sin\phi_k}{\omega_k^2}~d\gamma~,
  \label{e.dqn}
\end{align}
where $\omega_k = \sqrt{(\cos\phi_k-h)^2+\gamma^2\sin^2\phi_k}$ is the 
dispersion relation of the XY model. An extensive behaviour of 
\eqref{e.ds2cirac} is given by the continuous limit 
$\sum_k \to \frac{n}{2\pi}\int_0^{2\pi}d\phi$: if the resulting integral
is convergent, no superextensive behaviour can occur. However, from 
\eqref{e.dqn} it is clear that a possible (the only?) source of 
a divergent behaviour of $dq_k^2$ is the vanishing of the gap 
$\min_k \omega_k$. It is known that in the XY model this condition occurs only
for $h=1$, where one finds  for $\phi\simeq O(n^{-1})$ that 
$\min_k\omega\approx O(n^{-1})$.
Hence 
\begin{align}
  \max_k dq_k \approx O(n)\; dh + O(n^{-1})\;d\gamma~,
\end{align}
from which 
\begin{align}
  |g|\approx g_{hh} &= O(n^2)~, &\text{for }& h=1~.
\end{align}
On the other hand for $\gamma\to 0$, $\omega\simeq|h-\cos \phi|$, so if 
$h=\cos\phi+O(n^{-1})$ we obtain
\begin{align}
  d q_k\Big|_{\gamma\to0} = -2
  \frac{\phi_k}{(h-\cos \phi_k)}~d\gamma \simeq O(n)d\gamma~,
\end{align}
again recovering the scaling $|g|= O(n^2) $.

\end{document}